\RequirePackage{fix-cm}
\documentclass[smallcondensed]{svjour3}     % onecolumn (ditto)
\smartqed  % flush right qed marks, e.g. at end of proof

\usepackage{graphicx}
\usepackage{mathptmx}      % use Times fonts if available on your TeX system
\usepackage{mathtools}
\usepackage{amssymb}
\usepackage{enumitem}
\usepackage{todonotes}
\usepackage{hyperref}

\usepackage{algorithm}
\usepackage{algorithmicx}
\usepackage{algpseudocode}

\usepackage{tabularx}
\usepackage{tabulary}

\usepackage{placeins}

\hypersetup{
	pdfborder={0 0 0},
	colorlinks = false,
	pdflinkmargin=0pt,
}

\numberwithin{equation}{section}
\usepackage[sort]{natbib}
\newtheorem{dfn}[theorem]{Definition}

\newtheorem{rem}[theorem]{Remark}
\newtheorem{prp}[theorem]{Proposition}

\newtheorem{thm}[theorem]{Theorem}
\newtheorem{cor}[theorem]{Corollary}

\newcommand\orcidicon[1]{\href{https://orcid.org/#1}{\protect\includegraphics[width = .25cm]{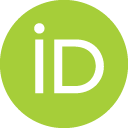}}}

\usepackage{bbding}

\newcommand{\define}{\coloneqq}

\newcommand{\algorithmiccontinue}{\textbf{continue}}

\journalname{}
\begin{document}
\title{A deterministic matching method 
		for exact matchings to compare the 
		outcome of different interventions		
\thanks{
		C.~Kirches acknowledges funding by Deutsche Forschungsgemeinschaft through Priority Programme 1962 ``Non-smooth and Complementarity-based Distributed Parameter Systems: Simulation and Hierarchical Optimization''.
		C.~Kirches was supported by the German Federal Ministry of Education and Research, grants n\textsuperscript{o} 61210304-ODINE, 05M17MBA-MoPhaPro and 05M18MBA-MOReNet.
	}
}

\titlerunning{Deterministic statistical exact matching}

\author{Felix Bestehorn$^{\textbf{1}}$\orcidicon{0000-0001-6339-2193} \and
	Maike Bestehorn$^{\textbf{2}}$	\and
	Christian Kirches$^{\textbf{1}}$\orcidicon{0000-0002-3441-8822}  %etc.
}

\authorrunning{F. Bestehorn et al.} % if too long for running head

\institute{\href{mailto: f.bestehorn@tu-bs.de}{\Envelope~}Felix Bestehorn 
			\and
			Christian Kirches\\
			\email{\{f.bestehorn, c.kirches,\}@tu-bs.de}\\
			\at
			~$^{\textbf{1}}$ {Institute for Mathematical Optimization, Technische Universit\"at Braunschweig, Braunschweig, Germany}
			\and 
			Maike Bestehorn\\
			\email{maike.bestehorn@t-online.de}\\
			\at
			~$^\textbf{2}${Sch\"aftlarn, Germany}
	}

\date{Received: date / Accepted: date}

\maketitle

\begin{abstract}
Statistical matching methods are widely used in the social and health sciences to 
estimate causal effects using observational data. 
Often the objective is to find comparable groups with similar covariate distributions in a dataset,
with the aim to reduce bias in a random experiment.
We aim to develop a foundation for deterministic methods which provide results 
with low bias, while retaining interpretability. 
The proposed method matches on the covariates and 
calculates all possible maximal exact matches
for a given dataset without adding numerical errors.
Notable advantages of our method over existing matching algorithms are 
that all available information for exact matches is used, no additional bias is introduced, 
it can be combined with other matching methods for inexact matching to reduce pruning
and that the result is calculated in a fast and deterministic way. For a given dataset the result
is therefore provably unique for exact matches in the mathematical sense. We provide proofs, instructions for 
implementation as well as a numerical example calculated for comparison on a complete survey.  
 %maybe C12 isntead
\MCS{62D20 \and 91B68}
\end{abstract}
\section{Introduction}
\label{sec:intro}

Statistical matching (SM) is widely used to reduce the effect of confounding
~\citep{Rubin1973,Anderson1980,Kupper1981} when estimating the causal effects of two different paths of action 
in an observational study. Such a study consists e.g. of a dataset containing two therapy groups $A$ and $B$, 
which in turn contain patients 
$a_1,\,\ldots,\,a_{\vert A \vert},\,b_1,\,\ldots,\,b_{\vert B \vert}$. 
Every patient $p\in D$ has a $s$-dimensional covariate vector $cv(p)$, 
describing the patients condition, and 
an observed value $\mathfrak{o}(p)$, describing the result of the therapy,
for examples see~\citep{Ray2012,Zhang2015,Gozalo2015,Zhang2016LeftIM,Cho2017,Burden2017,McEvoy2016,Schermerhorn2008,Lee2017,Capucci2017,Tranchart2016,Zangbar2016,Dou2017,Fukami2017,McDonald2017,Lai2016,Abidov2005,Adams2017,Kishimoto2017,Kong2017,Chen2016,Seung2008,Shaw2008,Liu2016,Svanstroem2013,Salati2017}. 

The goal in regard to SM would then be to find a matching such that
patients which are similar according to a chosen similarity measure, 
e.g.~Mahalanobis distance or propensity score, 
are compared with each other
and a reliable conclusion with regards to the preferable therapy under the circumstances
defined by the underlying model and hypothesis can be drawn from the matching, while bias potentially
introduced through a comparison of dissimilar patients is reduced. 
If the matching is also maximal in the 
sense that all possibly matchable patients, i.e. patients which are similar to other 
patients, are matched, the matching is called a maximal matching. 

Regrettably, minimizing bias is not 
the only key aspect to be considered: maximal matching can lead to the pruning of patients, thus possibly ignoring relevant information contained in the dataset. As these matchings are usually not unique, there can be a high variance in information in between matchings and thus conclusions drawn from them 
can potentially vary to a high degree. 
Hence one needs to find a matching optimized in regards to bias and pruning.
As the underlying distribution of the dataset and the influence of the therapy is
unknown, finding the optimal maximal patient-to-patient matching is difficult.
Based on the foundation 
laid by~\cite{RubinRosenbaum1983} for propensity score matching~(PSM), many 
different methods have been proposed to deal with this problem, e.g. nearest neighbour matching \citep{Rubin1972}, 
stratification matching on propensity scores~\citep{Anderson1980},
caliper matching~\citep{Stuart2010}, optimal matching~\citep{Rosenbaum1989}, coarsened matching~\citep{Iacus2012} or full matching~\citep{Hansen2004,Hansen2012}.
A comprehensive overview can be found for example in the article from~\cite{Stuart2010}. 

The aforementioned methods inspect either one or several matchings and try to cope with the problem of 
not being able to calculate all possible maximal matchings through numerical or stochastic methods (\citep{Stuart2010})
and have some limitations which have been investigated lately~\citep{King2019,Austin2011,Pearl2000}. 
Therefore, researchers find themselves sometimes in the difficult position where different matchings,
while being statistically sound, can suggest different conclusions.

Due to the aforementioned reasons we investigate a slightly different approach in this article.
After showing
that considering only one or several patient-to-patient matchings over the whole dataset
is inadequate as exponentially many different patient-to-patient matchings exist, implying that high variance 
in the deduced conclusions is possible,
we proceed to propose a method that matches clusters of patients.
The goal is to develop a method which uses all information contained in the dataset and considers all possible maximal 
matchings of a dataset in accordance to a chosen similarity, thus leading to low confounding and low variance. 
The proposed method has the desirable property of calculating a matching in accordance to the expectancy
value of all possible maximal exact matchings in the dataset, while being fast and deterministic and therefore can
support decision making processes 
as no additional errors are included during the matching process.

A short note on terminology: We use the terms 
terms therapy group, patient, covariate vector and observed value for clarity and simplicity 
of presentation and that they can be substituted
for any type of group, member of said group, properties of the member and observed result for the member.

\subsection{Contribution}
\label{sec:contr}
We investigate the quantity of possible exact matchings and show that, 
even under the restriction that only exact matches are considered, many possible 
matchings exists. For this reason we
propose a different approach, which uses all available information in a given dataset for an exact matching 
and show that the proposed method has desirable properties for SM. 
We confirm our theoretical contributions by evaluating a complete survey 
and comparing the proposed method with the de-facto standard for SM in such applications,
namely \textit{propensity score matching} (PSM).
\subsection{Structure of the Remainder}
\label{sec:struct}
We proceed to show that even in an exact matching context multiple possible matchings exist and propose
an algorithm to cope with this problem and prove desirable properties of the algorithm. 
In Section~\ref{sec:num_example}, we describe and confirm the findings made in the previous section 
based on a comparison of the proposed algorithm with an established method for SM on a complete survey. 
We conclude with Section~\ref{sec:conclusion}, which is used 
to summarize our results as well as the algorithm's benefits and drawbacks and to give an outlook 
for potential further development. 

\section{Deterministic exact matching}
\label{sec:exact:matching}

As stated in the introduction (Section~\ref{sec:intro}), the goal of SM in a general setting
is to match as many patients between two groups with regards
to a chosen similarity or distance measure as possible. 
One can immediately distinguish two cases:

\begin{itemize}
	\item Exact matching~\citep{Stuart2010,Iacus2012}: Only members of different sets with equal covariate vectors are matched.
	\item $\delta$-matching or caliper/inexact matching~\citep{Stuart2010}: Members of different sets can be matched if their 
	difference with regard to a chosen similarity measure is smaller than $\delta$.
\end{itemize}

Thus exact matching is a special case of $\delta$-matching for $\delta = 0$
and one can define potentially matchable patients in the following way.

\begin{dfn}[Matchable Patients]
	\label{def:matching_patients}
	Let $p$ and $q$ be two patients from different therapy \linebreak 
	groups and let $d(\,\cdot\,,\,\cdot)$ 
	be an arbitrary similarity measure. Then $p$ and $q$ are matchable patients for a $\delta$-matching,
	if $d(p,\,q) \le \delta$.
\end{dfn}

We will only consider exact matches, 
$\delta = 0$, for the remainder of this manuscript.
We refer to Section~\ref{sec:conclusion} for a discussion of a potential
extension of this method to $\delta$-matchings with $\delta > 0$.

We first formalize an 
observation made during statistical matching, which states that
the number of all different possible exact matches may be inhibitively large, so that not all
matchings can be computed in acceptable time. 
Based on this we proceed by showing that it is possible to calculate a 
matching on a dataset in accordance with the expectancy of the observed values, 
if one uses clusters instead of patients and show that the
proposed algorithm has additional desirable properties for statistical matching.

\subsection{Preparations}
\label{subsec:prep}

The general idea of the proposed method is to cluster patients with the same covariate vector 
for each therapy group and generate a matching between both therapy groups for the constructed clusters.

Clustering of patients $p$ and $q$ on their covariate vectors requires a similarity measure. 
In the remainder, 
we will use the $L_1$ distance measure (also known as Manhattan metric) 
\begin{equation}
\label{eq:manhattan}
d(p,\,q) \define \sum_{i = 1}^{s} \vert cv_i(p) - cv_i(q) \vert,
\end{equation}
but other distance measures (possibly defined through similarity measures) are applicable as well as long 
as they can be calculated for every pair of patients.

\begin{rem}
	Note that two patients $p$ and $q$ have identical covariate vectors if and only if $d(p,\,q) = 0$. Thus 
	for exact matching $d(p,\,q) = 0$ is required for all matchable patients $p$ and $q$. 
\end{rem}

Together with a distance measure we can define the notion of clusters.

\begin{dfn}[Cluster of patients]\label{def:cluster}
	In an SM context, a cluster of patients from one therapy group $H$ is a non-empty set $C_{H}$ of patients with properties
	\begin{enumerate}
		\item $d(p,\,q) = 0 \,\,\forall p,\,q \in C_{H}$. \label{item:def_cluster_1}
		\item For $p\in C_H$ it holds that $\nexists q \in H$ such that $q\notin C_H$ and $d(p,\,q) = 0$. \label{item:def_cluster_2}
		\item If $p\in C_H$, then the assigned covariate vector of $C_H$ is $cv(p)$. \label{item:def_cluster_3}
	\end{enumerate}
\end{dfn}

Hence, clusters have the following properties:

\begin{prp}\label{prop:cluster}
	Let $H$ be a therapy group in an SM context, then the following holds for clusters in this therapy group:
	\begin{enumerate}
		\item Every patient in $H$ belongs to exactly one cluster.
		\item Every cluster can have exactly one covariate vector assigned to it.
		\item Any two clusters in $H$ have different assigned covariate vectors.
	\end{enumerate}
\end{prp}

\begin{proof}
	We prove every characteristic individually:
	\begin{enumerate}
		\item As $d(p,\,p) = 0,\,\forall p \in H$, all patients belong to at least one cluster. Thus it remains to show that there exists no patient $p\in H$ belonging to two different clusters $C_{H,1}$ and $C_{H,2}$. Assume that $p\in C_{H,1}\cap C_{H,2}$ and let $q_1 \in C_{H,1}$ and $q_2 \in C_{H,2}$ be two patients in $C_{H,1}$ and $C_{H,2}$ respectively. Then it holds by Definition~\ref{def:cluster}.\ref{item:def_cluster_1} that ${d(p,\,q_1) = 0 = d(p,\,q_2)}$ and therefore $d(q_1,\,q_2)=0$. This is a contradiction to Definition~\ref{def:cluster}.\ref{item:def_cluster_2}. and therefore every patient belongs to exactly one cluster.
		\item As clusters are non-empty sets of patients every cluster has at least one covariate vector assigned to it. Therefore assume that cluster $C$ has two assigned covariate vectors $v_1$ and $v_2$ differing in at least one entry. Then by Definition~\ref{def:cluster}\ref{item:def_cluster_3} it holds that there exists patients $p,\,q \in C$ such that $v_1 = cv(p)$ and $v_2 = cv(q)$. As $v_1 \neq v_2$ holds by assumption it follows that $d(p,\,q) \neq 0$, contradicting Definition~$\ref{def:cluster}.\ref{item:def_cluster_1}$ as $p,\,q \in C$.
		\item Assume that different clusters $C_{H,1}$ and $C_{H,2}$ have the same assigned covariate vector. This implies that\linebreak
		$d(p,\,q) = 0,\,\forall p\in C_{H,1},\,q\in C_{H,2}$ and is a contradiction to Definition~\ref{def:cluster}.\ref{item:def_cluster_2}.
	\end{enumerate}
\end{proof}

Because of Proposition~\ref{prop:cluster}, clusters can be assigned unique covariate vectors. Thus the similarity of clusters $C_A$ and $C_B$ -- for therapy groups $A$ and $B$ respectively -- can be denoted similarly to patients as $d(C_A,\,C_B)$. This leads to the following observation in regards to exact matching:

For an arbitrary dataset $D = (A,\,B)$ and a cluster $C_H$ belonging to a therapy group ${H \in \{A,\,B\}}$ only the following situations can occur: 

\begin{enumerate}
	\item For $C_A$ there exists one cluster $C_B$ with $d(C_A,\,C_B) \equiv 0$. \label{item:situation_1}
	\item For $C_A$ there exists no cluster $C_B$ with $d(C_A,\,C_B) \equiv 0$.
	\item For $C_B$ there exists no cluster $C_A$ with $d(C_A,\,C_B) \equiv 0$.
\end{enumerate}

Only situation~(\ref{item:situation_1}) is relevant for exact matching of single patients or clusters, as exact matching can only occur for clusters with a corresponding cluster in the opposite therapy group. Furthermore if situation~(\ref{item:situation_1}) occurs, then the match is unique in regards to the clusters.

\begin{prp}\label{prop:cluster_equivalence}
	Let $C_A$ and $C_B$ be clusters from different therapy groups, then \\${d(C_A,\,C_B) \equiv 0}$ holds iff the two clusters have the same assigned covariate vector.
\end{prp} 

\begin{proof}
	Let $C_A$ and $C_B$ be clusters from different therapy groups and ${d(C_A,\,C_B) \equiv 0}$. As every cluster has exactly one assigned covariate vector it remains to show ${cv(C_A) \equiv cv(C_B)}$: 
	\begin{equation}
	\label{eq:cluster_eq}
	d(C_A,\,C_B) \equiv 0 \Leftrightarrow \sum_{i=1}^{s} \vert cv_i(C_A)-cv_i(C_B)\vert \equiv 0 \Leftrightarrow cv_i(C_A) \equiv cv_i(C_B),\,\forall 1\leq i \leq s.
	\end{equation}
	Thus both clusters have the same assigned covariate vector. The reverse direction follows as all implications in equation~\eqref{eq:cluster_eq} are given through equivalence.
\end{proof}

Motivated by the previous proposition and by the definition of matchable patients (Definition~\ref{def:matching_patients}) we can define exact matchable clusters.
\begin{dfn}[Matchable Cluster]
	\label{def:matching_cluster}
	Two clusters $C_A$ and $C_B$ of $A$ and $B$ respectively are exact matchable clusters \linebreak
	iff $d(C_A,\,C_B) \equiv 0$.
\end{dfn}

Equipped with Definitions~\ref{def:cluster} and~\ref{def:matching_cluster} as well as
Propositions~\ref{prop:cluster} and~\ref{prop:cluster_equivalence} we can identify 
the exact cardinality $n$ of exact matchable clusters. 

Additionally,
we can calculate the number of all different possible matchings between two exact matchable clusters 
(Proposition~\ref{prop:total_number}) and for whole datasets (Proposition~\ref{prop:group_number}).  
Note that the necessity to calculate all possible matchings
arises as observed results between different maximal matchings can show a large variation,
see Section~\ref{sec:num_example}, even if 
the dataset is large and matches are exact (Table~\ref{tab:results}).

\begin{prp}
	\label{prop:total_number}
	Let $C_A$ and $C_B$ be exact matchable clusters of $A$ and $B$ with $\vert C_A \vert = \mathfrak{a}$ and $\vert C_B \vert = \mathfrak{b}$ respectively, then
	\begin{enumerate}
		\item A set of exact matches $M$ between $A$ and $B$ with a maximal number of matches includes $\min(\mathfrak{a},\,\mathfrak{b})$ exact matches from $C_A$ and $C_B$.\label{item:total_number_1}
		\item For $C_A$ and $C_B$ there exists $\binom{\max(\mathfrak{a},\,\mathfrak{b})}{\min(\mathfrak{a},\,\mathfrak{b})}$ 
		sets with $\min(\mathfrak{a},\,\mathfrak{b})$ exact matches, where $\binom{\max(\mathfrak{a},\,\mathfrak{b})}{\min(\mathfrak{a},\,\mathfrak{b})}$ denotes the binomial coefficient of
		$\max(\mathfrak{a},\,\mathfrak{b})$ and $\min(\mathfrak{a},\,\mathfrak{b})$.\label{item:total_number_2}
	\end{enumerate}
\end{prp}

\begin{proof}
	We can w.l.o.g. assume that $\mathfrak{a} \le \mathfrak{b}$, as one can simply substitute $\mathfrak{a}$ for $\mathfrak{b}$ 
	and $\mathfrak{b}$ for $\mathfrak{a}$ in the other case. 
	Thus $\min(\mathfrak{a},\,\mathfrak{b}) = \mathfrak{a}$ and as 
	$C_A$ and $C_B$ are exact matchable clusters, $\mathfrak{a}$ matchable pairs
	$(p,\,q)$ with $p\in C_A,\,q\in C_B$ exist. Due to Proposition~\ref{prop:cluster_equivalence}, 
	this is the maximal matchable number of patients between $C_A$ and $C_B$
	as $M$ was assumed to be maximal.
	Assume now that only $a< \mathfrak{a}$ of these matchable pairs 
	are contained in $M$. This contradicts the maximality of $M$ as the remaining matchable pairs 
	could be added. This proves~(\ref{item:total_number_1}). 
	
	As $\mathfrak{a} \le \mathfrak{b}$, it holds that every element from $C_A$ is matched to elements of $C_B$
	and these matches constitute $\mathfrak{a}$ many tuples of matched elements, in short an $\mathfrak{a}$-tuple. 
	As $\vert C_B\vert = \mathfrak{b}$, one can construct 
	$\binom{\mathfrak{b}}{\mathfrak{a}}$ different matchings by selecting different elements of $C_B$ for 
	the matches.
	
	Therefore for two matching clusters $C_A$ and 
	$C_B$ exist $\binom{\max(\mathfrak{a},\,\mathfrak{b})}{\min(\mathfrak{a},\,\mathfrak{b})}$ 
	sets with $\min(\mathfrak{a},\,\mathfrak{b})$ exact matches in general.	
\end{proof}

\begin{prp}
	\label{prop:group_number}
	Let $n$ be the number of exact matching clusters in $A$ and $B$ and let 
	$C_{A,j}$ and $C_{B,j}$ with $1\le j \le n$ be exact matching clusters of $A$ and $B$ with 
	$\vert C_{A,\,j} \vert = \mathfrak{a}_j$ and $\vert C_{B,\,j} \vert = \mathfrak{b}_j$ respectively. Then the number of different 
	exact matchings that exist on the whole dataset is
	\begin{equation}
	\label{eq:group_number}
	\prod_{j=1}^{n} \binom{\max(\mathfrak{a}_j,\,\mathfrak{b}_j)}{\min(\mathfrak{a}_j,\,\mathfrak{b}_j)} 
	\end{equation}
\end{prp}

\begin{proof}
	By Proposition~\ref{prop:total_number} there exist $\binom{\max(\mathfrak{a}_j,\,\mathfrak{b}_j)}{\min(\mathfrak{a}_j,\,\mathfrak{b}_j)}$ 
	exact maximal matchings for every pair of exact matching clusters $1\le j \le n$. Therefore we have 
	$$\prod_{j=1}^{n} \binom{\max(\mathfrak{a}_j,\,\mathfrak{b}_j)}{\min(\mathfrak{a}_j,\,\mathfrak{b}_j)}$$ 
	maximal exact matchings in total.
\end{proof}

Note that the number of different maximal matchings in Proposition~\ref{prop:total_number}
is smaller for exact matchings than it is for 
$\delta$-matchings with $\delta>0$ or one-to-many matchings.

Proposition~\ref{prop:total_number} shows that even for small datasets
the naive way of calculating all possible matchings can be infeasible 
due to the binomial coefficient. One could now calculate only one or several matchings, but
these will possibly neglect important parts of available information in the dataset.

\subsection{A deterministic balancing exact matching algorithm}
\label{subsec:algorithm}

Motivated by the findings of the previous subsection and the use-oriented necessity to use all 
available information in a given dataset, we investigate the outcomes of an observed 
value in a dataset regarding clusters. For simplicity of presentation we henceforth assume that 
$\mathfrak{o}(x)$ is in $\{0,\,1\}$, see Remark~\ref{rem:obs_values} for a short discussion of other settings. 

\begin{dfn}[Relative frequency of an observed value in a cluster]
	\label{def:mean}
	Let \linebreak
	$C_H \define \{x_1,\,\ldots,\,x_{\vert C_H\vert }\}$ be a cluster in therapy group $H$. Then the relative frequency
	of the observed value $\mathfrak{o}(x_u) = 1$ in $C_H$ is defined as 
	\begin{equation}
	\label{eq:def_mean}
	F(C_H) \define \frac{1}{\vert C_H \vert}\sum_{u=1}^{\vert C_H\vert}\mathfrak{o}(x_u).
	\end{equation}
\end{dfn}

\begin{rem}[Different intervals for observational values]
	\label{rem:obs_values}
	Besides simplicity of presentation, the assumption that $\mathfrak{o}(x)$ is in $\{0,\,1\}$ has several advantages and
	can be easily generalized:
	\begin{enumerate}
		\item The relative frequency of the observed value $\mathfrak{o}(x_u) = 1$ in $C_H$ is $1-F(C_H)$.
		\item Any binary setting with $\tilde{\mathfrak{o}} \in \{0,\,K\},\,K\in\mathbb{R}$ can be mapped to $\mathfrak{o} \in\{0,\,1\}$.
		\item For non-binary outcomes $\mathfrak{o} \in \{K_1,\,K_2,\,\ldots\}$, one has to consider the modification 
		\linebreak
		$F^{(K_i)}(C_H) = \frac{1}{\vert C_H \vert}\sum_{u=1}^{\vert C_H\vert}\chi\big(\mathfrak{o}(x_u) = K_i\big)$, where $\chi\big(\mathfrak{o}(x_u) = K_i\big)$ denotes 
		the indicator 
		\linebreak 
		function, i.e. $\chi\big(\mathfrak{o}(x_u) = K_i\big) = \begin{cases}
		1, \quad \text{if } \mathfrak{o}(x_u) = K_i\\
		0, \quad \text{otherwise}
		\end{cases}$
	\end{enumerate} 
\end{rem}

The relative frequency of an observed value in a cluster 
can be seen as the relative outcome value 
for the cluster. In the context of statistical 
matching the observed value data of patients 
should only contribute to the final conclusion if 
the whole cluster can be matched. Regrettably this 
is not the case in general as a cluster $C_A \define \{x_1,\,\ldots,\,x_\mathfrak{a}\}$ 
does not 
necessarily have a matching cluster. Additionally even if a matching cluster
${C_B \define \{z_1,\,\ldots,\,z_\mathfrak{b}\}}$ for $C_A$ exists only
an accumulated observed value of $\min(\mathfrak{a},\,\mathfrak{b})$ 
patients should contribute to the frequency evaluated in the end to 
prevent a distortion of the end result by large clusters. 
A first approach, which we will refine subsequently, to prevent this 
distortion leads to the definition of the relative matching frequency of an observed value. 

\begin{dfn}[Relative matching frequency of an observed value]
	\label{def:rel_matching_mean}
	Let\\$C_A \define \{x_1,\,\ldots,\,x_\mathfrak{a}\},\,C_B \define \{z_1,\,\ldots,\,z_\mathfrak{b}\}$
	be two exact matching clusters. Then the relative matching frequency of an observed value $\mathfrak{o}(x_v) = 1$ for $C_A$ is defined as
	\begin{equation}
	\label{eq:def_rel_matching_mean}
	F_M(C_A) \define \frac{1}{\mathfrak{a}}\sum_{v=1}^{\min(\mathfrak{a},\,\mathfrak{b})}\mathfrak{o}(x_v).
	\end{equation}
\end{dfn}

Using the relative matching frequency of observed values to evaluate the final outcome of a dataset in terms of 
an observed value results in incomplete usage of information as only $\min\{\mathfrak{a},\,\mathfrak{b}\}$ 
patients are matched and thus only the observed values of $\min\{\mathfrak{a},\,\mathfrak{b}\}$ patients affect the outcome. 
This problem is independent of the matching method used, if the method does 
not consider all possible matchings. 
Note that many commonly used matching methods as described e.g. in \citep{Stuart2010} 
implicitly consider the relative matching frequency $F_M(C_A)$ as a result after a single matching 
realization as can be seen by interpreting the used 
patients as clusters appropriate to the chosen $\delta$.

We change the perspective to show that usage of the full information available is possible. For this we begin by considering 
one realization of a maximal matching between clusters as the result of a random experiment. As all patients 
in a cluster are the same with respect to their covariates, all patients in the same cluster should have the same probability to 
be chosen in a maximal matching to be matched to patients from a matching cluster. Thus every possible maximal matching has the same 
probability to appear in a single maximal matching experiment. Repeating the random experiment for maximal matchings between two clusters 
results in a sequence of maximal matchings, which we call a uniform sequence of matchings.

\begin{dfn}[Uniform sequence of matchings]
	\label{def:uniform_sequence}
	Let $C_A \define \{x_1,\,\ldots,\,x_\mathfrak{a}\}$ and \\$C_B \define \{z_1,\,\ldots,\,z_\mathfrak{b}\}$ 
	be two exact matching clusters. An infinite sequence of matchings $M = (M_1,\,M_2,\,\ldots)$ 
	is called a uniform sequence of matchings iff every possible matching between patients of $C_A$ and $C_B$ 
	has the same probability to be drawn as a matching $M_r$ in the sequence. 
\end{dfn}

With these notions we can show that the expectancy of the observed value over all possible maximal matchings is a term, 
whose value can be directly calculated through a cluster matching. Proposition~\ref{prop:mean_single} examines this for 
the case of two exact matching clusters.

\begin{prp}
	\label{prop:mean_single}
	Let $C_A \define \{x_1,\,\ldots,\,x_\mathfrak{a}\},\,C_B \define \{z_1,\,\ldots,\,z_\mathfrak{b}\}$
	be two exact matching clusters and let $M$ be a uniform sequence of maximal matched pairs between $C_A$ and $C_B$. Then
	\begin{enumerate}
		\item Every $M_k$ contains $\min(\mathfrak{a},\,\mathfrak{b})$ matching pairs and~\label{item:mean_single_number}
		\item The expectancy of the relative matching frequency over the sequence of uniform matchings for $C_A$ and $C_B$ can be calculated as
		\begin{equation}
		\label{eq:mean_single}
		\begin{aligned}
		\mathbb{E}({C_A}) &\define \lim\limits_{r\to \infty} \frac{1}{r}\Big(\sum_{k=1}^{r}F_M(C_A)_{k}\Big) = \frac{\min(\mathfrak{a},\,\mathfrak{b})}{\mathfrak{a}^2} \sum_{v=1}^{\mathfrak{a}}\mathfrak{o}(x_v)
		\quad \text{and} \\
		\mathbb{E}({C_B}) &\define \lim\limits_{r\to \infty} \frac{1}{r}\Big(\sum_{k=1}^{r}F_M(C_B)_{k}\Big) 
		= \frac{\min(\mathfrak{a},\,\mathfrak{b})}{\mathfrak{b}^2} \sum_{w=1}^{\mathfrak{b}}\mathfrak{o}(z_w). 
		\end{aligned}
		\end{equation}
	\end{enumerate}
\end{prp}

\begin{proof}
	By Proposition~\ref{prop:group_number}, every exact match with a maximum number 
	of matches includes $\min(\mathfrak{a},\,\mathfrak{b})$ pairs of $C_A$ and $C_B$. Therefore every $M_k$ contains
	$\min(\mathfrak{a},\,\mathfrak{b})$ pairs of patients from $C_A$ and $C_B$. 
	For the second part we can w.l.o.g. assume that $\mathfrak{a}\le \mathfrak{b}$, as one can simply substitute in the other case.
	As $\mathfrak{a} = \min(\mathfrak{a},\,\mathfrak{b})$ it follows that $$F_M(C_A)_k = F_M(C_A) = F(C_A) = \frac{1}{\mathfrak{a}}\sum_{v=1}^{\mathfrak{a}} \mathfrak{o}(x_{v}) = \frac{\mathfrak{a}}{\mathfrak{a}^2}\sum_{v=1}^{\mathfrak{a}} \mathfrak{o}(x_{v})  $$ for all $k$. 
	
	For $\mathfrak{a} = 1$ it follows that exactly one patient $\widetilde{z}_{k,\,w},\,w\in\{1,\,\ldots,\,\mathfrak{b}\}$ of $C_B$ gets chosen in every maximal matching $M_k$. As all patients have the same probability to be chosen the probability is $\frac{1}{\mathfrak{b}}$ for every patient. Evaluating the limits and referring to the patients of $C_B$ chosen in one realization of a maximal matching as
	$\widetilde{z}_{k,\,w}$ leads to:
	\begin{equation}
	\label{eq:group_number_1}
	\lim\limits_{r\to \infty} \frac{1}{r}\Big(\sum_{k=1}^{r}F_M(C_B)_{k}\Big) =
	\lim\limits_{r\to \infty} \frac{1}{r}\Big(\sum_{k=1}^{r} \frac{1}{\mathfrak{b}}\sum_{v=1}^{\mathfrak{a}}\mathfrak{o}(\widetilde{z}_{k,\,w})\Big) = 
	\frac{1}{\mathfrak{b}^2}\sum_{w=1}^{\mathfrak{b}} \mathfrak{o}(z_{w}),
	\end{equation}
	where the second equation follows by the law of large numbers as all patients have the same probability to be chosen in a maximal matching.
	
	Now let $\mathfrak{a} > 1$. Thus out of $\mathfrak{b}$ patients $\mathfrak{a}$ patients get matched and every patient has the same probability to be chosen in a maximal matching $M_k$. Again by the law of large numbers in the second equation it follows that
	\begin{equation}
	\lim\limits_{r\to \infty} \frac{1}{r}\Big(\sum_{k=1}^{r} F_M(C_B)_k\Big) = 		
	\lim\limits_{r\to \infty} \frac{1}{r}\Big(\sum_{k=1}^{r} \frac{\mathfrak{a}}{\mathfrak{b}}\sum_{v=1}^{\mathfrak{a}}\mathfrak{o}(\widetilde{z}_{k,\,w})\Big) = 
	\frac{\mathfrak{a}}{\mathfrak{b}^2}\sum_{w=1}^{\mathfrak{b}} \mathfrak{o}(z_{w}).
	\end{equation}
\end{proof}

It is straightforward to generalize Proposition~\ref{prop:mean_single} to uniform 
sequences of matchings over therapy groups containing several clusters, 
as a single realization of maximal matchings between clusters 
is independent of maximal matchings between other clusters.

\begin{prp}
	\label{prop:mean_total}
	Let $n$ be the number of exact matching clusters and let $C_{A,j}$ and $C_{B,j}$ with $1\le j \le n$ be exact matching clusters of $A$ and $B$
	with $\vert C_{A,\,j} \vert = \mathfrak{a}_j$ and ${\vert C_{B,\,j} \vert = \mathfrak{b}_j}$ respectively. Let $M$ be the uniform sequence of maximal exact
	matchings between all clusters. Then
	\begin{enumerate}
		\item Every element $M_k$ of $M$ contains 
		\begin{equation}
		\label{eq:mean_total_elements}
		\vert M_k \vert = \sum_{j=1}^{n}\min(\mathfrak{a}_j,\,\mathfrak{b}_j) 
		\end{equation}
		matches.
		\item For the expectancy of the relative matching frequencies of observed values for $A$ of any maximum exact matching it holds that
		\begin{equation}
		\label{eq:mean_total}
		\mathbb{E}_{A} \define \lim\limits_{r \to \infty} \frac{1}{r} \sum_{k=1}^{r}\Big(\sum_{j = 1}^{n} F_M(C_{A,\,j})_k\Big) = \sum_{j=1}^{n} \mathbb{E}(C_{A,\,j}),
		\end{equation}
		with $\widetilde{x}_{j,\,k,\,v},\,{v\in\{1,\,\ldots,\,\mathfrak{a_j}\}}$ as the patients from cluster $C_{A,\,j}$ chosen in the $k$-th matching $M_k$. Analogously it holds that 
		\begin{equation*}
		\mathbb{E}_{B} \define \lim\limits_{r \to \infty} \frac{1}{r} \sum_{k=1}^{r}\Big(\sum_{j = 1}^{n} F_M(C_{B,\,j})_k\Big) = \sum_{j=1}^{n} \mathbb{E}(C_{B,\,j}),
		\end{equation*}
	\end{enumerate}
\end{prp}

\begin{proof}
	Equation~\eqref{eq:mean_total_elements} holds as it is a summation over 
	the equation from Proposition~\ref{prop:mean_single}.\ref{item:mean_single_number}. 
	Similarly equation~\eqref{eq:mean_total} holds as a summation over equation~\eqref{eq:mean_single} as 
	a maximal matching between two clusters is independent of a maximal matching between other clusters.
\end{proof}

The previous proposition shows that for each therapy group $A$ and $B$ not all possible matches are realized during a
single matching and that the relative matching frequency
of observed values for uniform sequences of matchings converges to the expectancy of the
observed results, which in this case is an easily calculable value. The added benefit being that the
term is unique for
a dataset and independent of the used matching method.\\ 

We are now prepared to propose an algorithm, which calculates the expectancy of the relative matching 
frequencies of observed values in a deterministic fashion. The full algorithm (Algorithm~\ref{alg:DBSeM}) is
divided into three stages.

In its first stage clusters according to Definition~\ref{def:cluster} are generated (Algorithm~\ref{alg:cluster}). In the second
stage the algorithm will try to match as many clusters as possible  (Algorithm~\ref{alg:matching}), while in the
third stage it weights the patients of each cluster in accordance to the size of its
matching cluster and its own size according to equation \eqref{eq:mean_single} (Algorithm~\ref{alg:weighting}).

\begin{algorithm}[H]
	\caption{Clustering step}
	\label{alg:cluster}
	\begin{algorithmic}[1]
		\State Set $c = 0$ and $is\_clustered(x_v) =0$ for all patients in $A$. \label{state:cluster_1}
		\For{each patient $x_v,\, 1\leq v \leq \vert A \vert $} \label{state:cluster_2}
		\If{$is\_clustered(x_v) \equiv 0$} \label{state:cluster_3}
		\State Set $c = c+1$, $C_{A,\,c} := \{x_v\}$ and $is\_clustered(x_v) = 1$. \label{state:cluster_4}
		\Else
		\State \algorithmiccontinue
		\EndIf \label{state:cluster_5}
		\For{each patient $x_u$ with $v < u\leq \vert A\vert $ and $is\_clustered(x_u) \equiv 0$} \label{state:cluster_6}
		\If{$d(x_u,\,C_{A,\,c})\equiv 0$} \label{state:cluster_7}
		\State Set $C_{A,\,c} = C_{A,\,c} \cup x_u$ and $is\_clustered(x_u) = 1$ \label{state:cluster_8}
		\EndIf \label{state:cluster_9}
		%		\algstore{clustering_alg}
		%	\end{algorithmic}
		%\end{algorithm}
		%%
		%%%	\floatstyle{nocaptionruled}
		%%%	\restylefloat{algorithm}
		%%	\begin{algorithm}[H]
		%%		\begin{algorithmic}[1]
		%%			\algrestore{clustering_alg}
		\EndFor \label{state:cluster_10}
		\EndFor \label{state:cluster_11}
		\State Set $n_A = c$.\label{state:cluster_12}
		\State Repeat steps \ref{state:cluster_1} -- \ref{state:cluster_11} for $B$ and set $n_B = c$. \label{state:cluster_13}
		\State \Return $C_{A,\,1},\,\ldots,\,C_{A,n_A},\,C_{B,\,1},\,\ldots,\,C_{B,n_B}$.
	\end{algorithmic}
\end{algorithm}

\begin{algorithm}[H]
	\caption{Matching step}
	\label{alg:matching} 
	\begin{algorithmic}[1]
		\State Set $i = 1$.
		\For{every cluster $C_{A,\,g},\,1\leq g \leq n_A$} \label{state:M_1}
		\State Search for cluster $C_{B,\,i}$ with $d(C_{A,\,g},\,C_{B,\,i}) \equiv 0$. \label{state:M_2}
		\If{a cluster $C_{B,\,i}$ was found in the previous step} \label{state:M_4}
		\State Create matching set $M_i = \emptyset$. \label{state:M_3}
		\State Set $M_i = \{C_{A,\,g},\,C_{B,\,i}\}$. \label{state:M_5}
		\State Set $i = i+1$.
		\EndIf \label{state:M_6}
		\EndFor \label{state:M_7}
		\State \Return Matching sets $M_1,\,\ldots,\,M_n$.
	\end{algorithmic}
\end{algorithm}
\begin{algorithm}[H]
	\caption{Weighting step}
	\label{alg:weighting}
	\begin{algorithmic}[1]
		\State Set $w(C_{A,\,g}) = 0\,\,\forall 1\leq g \leq n_A$ and $w(C_{B,\,h}) = 0\,\,\forall 1\leq h \leq n_B$
		\For{all $C_{A,\,g},\,1\leq g \leq n_A$, with $M_g \neq \emptyset$} \label{MW:state_2}
		\State Search for $C_{B,\,h},\,1\leq h \leq n_B$, as the previously calculated matching cluster of $C_{A,g}$.
		\State Calculate $S_{A,\,g} := S_{B,\,h} := \min\{\vert C_{A,\,g}\vert,\,\vert C_{B,\,h}\vert \}$.
		\State Compute $w(C_{A,\,g}) := S_{A,\,g}/\vert C_{A,\,g} \vert^{2}$ and $w(C_{B,\,h}) := S_{B,\,h}/\vert C_{B,\,h} \vert^{2}$.
		\EndFor
		\State Compute Min-weighted results: \label{MW:state_4}
		\begin{eqnarray}
		R_A &:=& 
		\sum_{h=1}^{n_A} \big[ w(C_{A,\,h}) \sum_{v=1}^{\vert C_{A,\,h}\vert} \mathfrak{o}(x_{v,\,h})\big],\label{eq:algorithm_A}\\ 
		R_B &:=& 
		\sum_{g=1}^{n_B} \big[w(C_{B,\,g}) \sum_{w=1}^{\vert C_{B,\,g}\vert} \mathfrak{o}(z_{w,\,g})\big]\label{eq:algorithm_B},
		\end{eqnarray} where $x_{v,\,h} \in C_{A,\,h}$ and $z_{w,\,g} \in C_{B,\,g}$.
		\State \Return weighted results $R_A,\,R_B$.
	\end{algorithmic}
\end{algorithm}
%\newpage
Linked together, Algorithms~\ref{alg:cluster} -- \ref{alg:weighting} form the~\hyperref[alg:DBSeM]{DeM algorithm}.
\begin{algorithm}[H]
	\caption{Deterministic balancing score exact matching algorithm (DeM)}
	\label{alg:DBSeM}
	\begin{algorithmic}[1]
		\State Cluster the patients with Algorithm~\ref{alg:cluster} for $A$ and $B$
		into $C_{A,\,1},\,\ldots,\,C_{A,n_A},\,C_{B,\,1},\,\ldots,\,C_{B,n_B}$.
		\State Compute matchings sets $M_1,\,\ldots,\,M_n$ through application of 
		the matching Algorithm~\ref{alg:matching} on the clusters $C_{A,\,1},\,\ldots,\,C_{A,n_A},\,C_{B,\,1},\,\ldots,\,C_{B,n_B}$.
		\State Compute the weighted result with  Algorithm~\ref{alg:weighting} for $M_1,\,\ldots,\,M_n$
		\State \Return Weighted results $R_A,R_B$ and the set of matched clusters $M$.
	\end{algorithmic}
\end{algorithm}

As shown in Proposition~\ref{prop:mean_total}, the~\hyperref[alg:DBSeM]{DeM algorithm} calculates the expectancy value for the uniform sequence of exact maximal matchings for a given dataset (equations \eqref{eq:mean_total} and respectively~\eqref{eq:algorithm_A} or~\eqref{eq:algorithm_B}), and uses every information contained in the dataset available for an exact matching  (equations~\eqref{eq:mean_total_elements} and steps~\ref{state:M_4} -- \ref{state:M_6} of Algorithm~\ref{alg:matching} in conjunction with Algorithm~\ref{alg:cluster}). This is summarized in the following theorem.

\begin{thm}[Matching properties of the DeM algorithm]
	\label{thm:maximal_matches} ~~\linebreak
	The~\hyperref[alg:DBSeM]{DeM algorithm}, Algorithm~\ref{alg:DBSeM}, 
	\begin{enumerate}
		\item matches all possible exact matches and
		\item produces exactly one matching result in accordance to the expectancy value of all possible matches in the dataset.
	\end{enumerate}
\end{thm}

Additionally one can prove that the proposed~\hyperref[alg:DBSeM]{DeM algorithm}  (Algorithm~\ref{alg:DBSeM}) is fast and deterministic. 

\begin{thm}
	\label{thm:runtime}
	The~\hyperref[alg:DBSeM]{DeM algorithm} (Algorithm~\ref{alg:DBSeM}) is a deterministic algorithm and has 
	a runtime of $O(\vert A \vert \cdot \vert B \vert\cdot s + \vert A \vert^2 + \vert B \vert^2)$, where $s$ is 
	the dimension of the covariate vectors.
\end{thm}

\begin{proof}
	An algorithm is deterministic if given a particular input it will always produce the same output. 
	Algorithm~\ref{alg:DBSeM} takes a dataset as input and
	will, in the case of an exact matching, always produce the same clusters during step \ref{alg:cluster}. 
	As the same clusters were produced in step \ref{alg:cluster}, the same clusters are matched in step 
	\ref{alg:matching}, because of Proposition~\ref{prop:cluster_equivalence}. Step \ref{alg:weighting} 
	calculates the weight of the respective matched clusters, which is always same since the matched clusters
	from step \ref{alg:matching} are the same. Thus Algorithm~\ref{alg:DBSeM} is deterministic.
	
	Evaluating the runtime can be achieved by looking at every step separately:
	\begin{enumerate}
		\item In Algorithm~\ref{alg:cluster} the steps~\ref{state:cluster_1}--\ref{state:cluster_12} have a runtime of $\vert A \vert^2$, while step \ref{state:cluster_13} has a runtime of $\vert B \vert^2$. \label{item:runtime_1}
		\item Algorithm~\ref{alg:matching} investigates every cluster in $B$ at most $\vert A \vert$ times and every comparison between clusters needs $s$ operations to determine the distance. Thus Algorithm~\ref{alg:matching} has a runtime of $O(\vert A \vert \cdot \vert B \vert\cdot s)$. \label{item:runtime_2}
		\item As $n_A \le \vert A \vert$ and $n_B \le \vert B \vert$ it follows that Algorithm~\ref{alg:weighting} has a runtime of $O(\max\{\vert A\vert,\,\vert B\vert\})$. \label{item:runtime_3}
	\end{enumerate}
	Adding all the runtimes together and making no further assumptions in regards to the comparative size of $s,\,\vert A \vert$ and $\vert B \vert$, one concludes that Algorithm~\ref{alg:DBSeM} has a runtime of  $O(\vert A \vert \cdot \vert B \vert\cdot s + \vert A\vert^2 + \vert B \vert^2)$.
\end{proof}

The exclusive applicability of the proposed algorithm to exact matches which can be seen as a limitation will be discussed in Section~\ref{sec:conclusion}.

\subsection{Additional properties of the DeM algorithm}
\label{subsec:properties}

As shown in the previous subsection, the
~\hyperref[alg:DBSeM]{DeM algorithm}
calculates matchings in accordance to the expected value over 
all possible matchings in the dataset. This section discusses two additional properties of the \hyperref[alg:DBSeM]{DeM algorithm}.\\

We first discuss an a posteriori property of the proposed cluster matching. 
For statistical tests it is often necessary to calculate the variance inherent to the final matching result. For clusters this can be achieved by 
looking at matchings through the perspective of hypergeometric distributions:

Let $C_A \define \{x_1,\,\ldots,\,x_\mathfrak{a}\},\,C_B \define \{z_1,\,\ldots,\,z_\mathfrak{b}\}$
be two exact matching clusters, then $\mathfrak{a}$ can be interpreted as the population number of which $\sum_{v=1}^{\mathfrak{a}}\mathfrak{o}(x_v)$
have some property and $\min(\mathfrak{a},\,\mathfrak{b})$ of $\mathfrak{a}$ patients are chosen in this maximal matching. Thus 
a realization of a maximal matching in the sense of relative matching frequencies can be interpreted as a sample drawn from a hypergeometrically distributed random variable
projected onto the interval $[0,\,1]$. 

The view of maximal matchings as realization of a drawing from a hypergeometric distribution concurs with the results of Propositions~\ref{prop:mean_single} and~\ref{prop:mean_total} from the previous subsection as the expectancy of a hypergeometric distribution for two exact matching clusters is $\mathbb{E}(C_A)$ and $\mathbb{E}(C_B)$, respectively, where the terms $\mathfrak{a}$ and $\mathfrak{b}$ stem from reversing the normalization done in the previous subsection. 

Taking this perspective allows to calculate the variance for maximal matchings.

\begin{prp}
	\label{prop:variance}
	Let $C_A \define \{x_1,\,\ldots,\,x_\mathfrak{a}\},\,C_B \define \{z_1,\,\ldots,\,z_\mathfrak{b}\}$
	be two exact matching clusters. Then the variance of matchings for $C_A$ is given by
	\begin{equation}
	\label{eq:single_variance}
	\mathrm{Var}(C_A) = \mathbb{E}(C_A)\Big(1-\frac{\sum_{v=1}^{\mathfrak{a}}\mathfrak{o}(x_v)}{\mathfrak{a}}\Big)\frac{\mathfrak{a}- \min(\mathfrak{a},\,\mathfrak{b})}{\mathfrak{a}-1}
	\end{equation}
\end{prp}

\begin{proof}
	Viewing one realization of a cluster matching as the realization of a hypergeometrically distributed random variable yields the probability of picking one patient with $\mathfrak{o}(x_v) \equiv 1$ as
	$\frac{\sum_{v=1}^{\mathfrak{a}}\mathfrak{o}_(x_v)}{\mathfrak{a}}$. From Proposition~\ref{prop:mean_single} it is know that
	$\mathbb{E}(C_A) = \frac{\min(\mathfrak{a},\,\mathfrak{b})}{ \mathfrak{a}^2} \sum_{v=1}^{\mathfrak{a}}\mathfrak{o}(x_v)$. Thus using the formula for the variance of
	hypergeometric distributions yields equation \eqref{eq:single_variance}.
\end{proof}

As exact matchings of different clusters are independent from each other, the variance
for a matching over a therapy group follows immediately from the previous Proposition~\ref{prop:variance}.

\begin{cor}
	\label{cor:total_variance}
	Let $n$ be the number of exact matching clusters in $A$ and $B$ and let $C_{A,j}$ and $C_{B,j}$ with $1\le j \le n$ be exact matching clusters of $A$ and $B$
	with $\vert C_{A,\,j} \vert = \mathfrak{a}_j$ and $\vert C_{B,\,j} \vert = \mathfrak{b}_j$, respectively. Let $M$ be the uniform sequence of maximal exact
	matchings between all clusters.	Then the variance of therapy group $A$ is
	\begin{equation}
	\label{eq:total_variance_A}
	\mathrm{Var}(A) = \sum_{j=1}^{n}\mathrm{Var}(C_{A,\,j}).
	\end{equation}
	For therapy group $B$ equation~\ref{eq:total_variance_A} holds similarly.
\end{cor}

Note that all values used in Proposition~\ref{prop:variance} and Corollary~\ref{cor:total_variance} 
are available after matching with the \hyperref[alg:DBSeM]{DeM algorithm} and that the clusters are matched such that
no additional factor is introduced into the variance. 
Additionally note that the variance of a cluster for which
all patients are matched, i.e. for $C_A$ with $\min(\mathfrak{a},\,\mathfrak{b}) = \mathfrak{a}$, is $0$. The same holds for
clusters for which all patients have the same observed result, as then 
either $F_w(C_A) = 0$ or  {$\Big(1-\frac{\sum_{v=1}^{\min(\mathfrak{a},\,\mathfrak{b})}\mathfrak{o}(x_v)}{\mathfrak{a}}\Big) = 0$}. Thus at least half of all matched clusters from $A$ and $B$ fulfil either the $\min(\mathfrak{a},\,\mathfrak{b}) = \mathfrak{a}$ or $\min(\mathfrak{a},\,\mathfrak{b}) = \mathfrak{b}$ condition, and therefore have a variance of $0$ in the matching calculated by the DeM algorithm.

The second property  we discuss relates to the calculation of clusters and the initial matching procedure. 
\cite{RubinRosenbaum1983} defined the balancing score $b(p)$ of a patient $p$ as a value assignment 
such that the conditional distribution of $cv(p)$ is the same for patients $p$ from both treatment
groups, $A$ and $B$. They have shown that $cv(p)$ is the finest balancing score (\citep{RubinRosenbaum1983}, section $2$) and 
that if
treatment assignment is strongly ignorable, then the difference between the two respective treatments 
is an unbiased estimate of the average treatment effect at that balancing score value (\citep{RubinRosenbaum1983},
theorem $3$). Since we use $cv(p)$ in our calculations, the 
result calculated by Algorithm~\ref{alg:DBSeM} is an unbiased estimate of the average
treatment effect, if the strong ignorability assumption holds,
additionally to the properties proven previously.

\section{Numerical example}
\label{sec:num_example}

\subsection{Description of dataset and setup}
\label{subsec:setup}

We use an official complete survey
to illustrate the effect of ignoring different possible matchings as well as the results of the proposed DeM algorithm.  

The dataset used is the quality assurance dataset of isolated aortic valve procedures in $2013$, which is an official mandatory dataset including all
aortic valve surgery cases in German hospitals. It contains patient information (covariates) and mortality information (observed results) for $17,427$ patients. For each patient the corresponding record contains $s=19$ covariate variables. The $17,427$ patients are divided into two therapy groups. $9,848$ SAVR cases (replacement surgery of aortic valves) and $7,579$ TF-AVI cases (transcatheter/transfemoral implantation of aortic valves). The cases were documented in accordance with \S 137 Social Security Code V (SGB V) by hospitals registered under \S 108 SGB V. The data collection is compulsory for all in-patient isolated aortic valve procedures in German hospitals. The dataset is held by the Federal Joint Committee (Germany) and freely accessible for researchers after application. Given this dataset, it can be safely assumed that the data is independent in a statistical sense as patients were only recorded once.

We proceed to compare the proposed DeM algorithm with two other approaches: the de-facto standard for statistical matching, the $1$:$1$ propensity score matching (PSM), as well as a bootstrapped variant of $1$:$1$ PSM by~\cite{Austin2014}. 
For the regression based PSM algorithms, relevant regression variables and their values have to be determined. For our example, we consider the $H_0$-hypothesis: \textit{The mortality-rate does not depend on therapy}, for which the relevant variables are internationally validated in the Euroscore II (\url{http://www.euroscore.org}). The corresponding regression values for this setting are taken from the quality assurance dataset of isolated aortic valve procedures. PSM itself was then calculated using functions provided by IBM SPSS Statistics for Windows, Version $24.0$.

The decision to use $1$:$1$ PSM for comparison was made as it has the highest amount of possible matchings for fixed match-sizes and is the most commonly used~\citep{Stuart2010}. Furthermore all possible $\mathfrak{v}:\mathfrak{w}$ matchings, for arbitrary $\mathfrak{v},\,\mathfrak{w} \in \mathbb{N}$ are included in the set of possible $1$:$1$ matchings, while the reverse is obviously not true for arbitrary $\mathfrak{v},\,\mathfrak{w} \in \mathbb{N}$ and any given dataset.

We additionally note that a match of two patients with $\delta >0$ in this dataset would imply a difference of at least $5\%$ between patients in regards to covariates as there are only $19$ covariate variables. 

We explicitly stress that the purpose of this section is not the recommendation of any kind of treatment, but the illustration of the usage of results presented in Subsections~\ref{subsec:prep}--\ref{subsec:properties}.

\subsection{Computational results}
\label{subsec:results}

We computed the exact $1$:$1$ matchings with PSM and the proposed \hyperref[alg:DBSeM]{DeM algorithm}. For comparison 
purposes we present seven realizations of the non-deterministic PSM \text{(Set 1 -- 7)}.
Out of the $9,848$ SAVR patients $3,361$ had at least one exact TF-AVI match, while out of the $7,579$ TF-AVI patients, $2,249$ patients had at least one exact SAVR match. Thus one third of all patients could be exactly matched. As stated in the previous section, the null hypothesis for calculation of 
the p-values was $H_0:$ \textit{The mortality-rate does not depend on therapy}.
The results of some maximum matchings and the differences between them are indicated in Table~\ref{tab:results}. 

\begingroup
\renewcommand*{\thefootnote}{\alph{footnote}}
\begin{table}
	%	\captionsetup{justification=centering}
	\caption{Results for maximal matchings}
	\label{tab:results}
	\begin{tabularx}{\linewidth}{l|ll|ll|l}
		$1{,}502$ exact matchings with & \multicolumn{2}{|c|}{SAVR} & \multicolumn{2}{|c|}{TF-AVI} &$\chi^2$ Test \\
		regards to all $19$ Euroscore II& \multicolumn{2}{|c|}{in-hospital death} & \multicolumn{2}{|c|}{ in-hospital death} & (2-tailed)\\
		variables and without replacement& count & \% & count & \% & p-value \\
		\label{table:set_1}PSM Set 1 & $73$ & $4.9\%$ & $33$ & $2.2\%$ & $<0.0001$\\
		\label{table:set_2}PSM Set 2 & $73$ & $4.9\%$ & $34$ & $2.3\%$ & $<0.0001$\\
		\label{table:set_3}PSM Set 3 & $42$ & $2.8\%$ & $32$ & $2.1\%$ & $0.2398$\\
		\label{table:set_4}PSM Set 4 & $24$ & $1.6\%$ & $15$ & $1.0\%$ & $0.1470$\\
		\label{table:set_5}PSM Set 5 & $73$ & $4.9\%$ & $50$ & $3.3\%$ & $0.0342$\\
		\label{table:set_6}PSM Set 6 & $24$ & $1.6\%$ & $50$ & $3.3\%$ & $0.0021$ \\ 
		\label{table:set_7}PSM Set 7 & $73$ & $4.9\%$ & $15$ & $1.0\%$ & $0.0001$ \\
		\label{table:bootstrap}Uniform Bootstrapping ($10{,}000$ samples)& $52.47$ & $3.49\%$ & $32.10$ & $2.14\%$ & $0.0210$ (t-test)\footnotemark\\
		
		\label{table:DBSeM}\hyperref[alg:DBSeM]{\textbf{DeM}} & $\textbf{53.01}$ & $\textbf{3.5\%}$& $\textbf{32.32}$ & $\textbf{2.1\%}$ &  $\textbf{0.0227}$\\
	\end{tabularx}
\end{table}
\footnotetext{The t-test values for all sets without replacement are $<0.0001$, with replacement $0.0005$.}	
\endgroup

A maximal matching $1$:$1$ comprises $1,502$ matching pairs of patients, meaning that 
$3,004$ patients were matched during any exact non-cluster matching. 
We only considered maximal exact matches, therefore every shown matching matches the maximal 
number of patients possible and matches two patients if and only if their covariates are equal,
meaning that the standardized differences in all presented sets is $0$. 
Still the large discrepancy between the observed results in the shown sets 
immediately indicates that many maximal matchings exist, as 
shown in Proposition~\ref{prop:group_number}.
Calculating all possible maximal matchings would be a futile endeavour and
not necessary if the observed results between different maximal matchings would not vary. 
Unfortunately observed results can vary to a very high degree, as can be seen in Table~\ref{tab:results}. 
They vary in such a way that one could even draw different conclusions based on the matching one 
calculated, see sets~\hyperref[table:set_1]{$1$},~\hyperref[table:set_6]{$6$} and~\hyperref[table:set_7]{$7$}, 
while arguing that the calculated p-value is below a threshold of $1\%$. 
For other sets one can see that they are on either side of the spectrum, favouring one, the other, or no therapy.
Even the \hyperref[table:bootstrap]{bootstrapping result} for $10,000$ samples did not exhaust all possible maximal
matchings and is only similar to the \hyperref[table:DBSeM]{result of DeM}, which gives as a result the
expectancy of all possible maximal exact matches in the given dataset, see Theorem~\ref{thm:maximal_matches}. 

As can be seen from the results, given one dataset and a non-deterministic method 
one could obtain different results even when regression variables are given, which leads to
uncertainty in the evaluation process as fellow researchers cannot reconstruct results obtained 
through statistical matching based on regression methods. The proposed algorithm tries to 
resolve this issue for exact matches. Even though this is a limitation in applicability, exact
cluster matches obtained through the \hyperref[alg:DBSeM]{DEM algorithm} can be used at the
core of a matching, ascertaining that at least the exact matchable contingent of a dataset
is matched deterministically (Theorem~\ref{thm:runtime}) and corresponds to the expectancy 
value of the exact matches (Theorem~\ref{thm:maximal_matches}).
Additionally, if in large datasets no exact matches can be found, 
researchers should thoroughly investigate for systematic differences in the therapy groups,
as comparisons of therapy-effects are not recommended if systematic differences exist.

\section{Conclusion}
\label{sec:conclusion}

We proposed an alternative deterministic exact matching method (\hyperref[alg:DBSeM]{DeM}) for SM
in the exact case. The proposed method is based
on matching clusters of patients from therapy groups instead of matching patients to patients directly. 
The presented cluster matching approach computed with the \hyperref[alg:DBSeM]{DeM algorithm} (Algorithm~\ref{alg:DBSeM})
extracts all possible information from a given dataset as all possibly 
matchable patients get matched and the constructed matching is in 
accordance with the expectancy value of the dataset (Theorem~\ref{thm:maximal_matches}). 
The constructed matching also has the desirable property of having low variance 
while being in accordance to the expectancy of all possible maximal exact matchings in the dataset 
(Proposition~\ref{prop:variance} and Corollary~\ref{cor:total_variance}).

As the proposed algorithm is deterministic and fast (Theorem~\ref{thm:runtime}) as
well as easy to implement, it can be used to produce exact matchings on 
datasets and to discuss findings in a reliable way as the results are 
easily reproducible. This is 
an important property as it makes a subsequent decision-making process more transparent 
and not susceptible to random events, such as 
random draws not in accordance to the expectancy.
Thus discussions about conclusions drawn can be done based on the dataset and the method used 
for data acquisition as 
uncertainties regarding the matching method are eliminated through a proven guarantee
that there are no additional errors introduced by the matching
procedure.

The results from the numerical example, calculated on a dataset 
containing a complete survey, validate the shown theoretical propositions and theorems. The 
proposed method can furthermore be seen as an extension of state 
of the art methods as results obtained through their usage would converge in the limit  
against the result calculated through the proposed algorithm. 

The exclusive applicability of the proposed algorithm to exact matches might be seen as a limitation. 
Then again for small datasets, which are statistically more prone to high variance in regards 
to two different matchings, the proposed algorithm provides a reliable result for the exact matches.
For the case of large datasets, a practitioner should be wary if only few exact matches exist 
or a matching result varies to a high degree from the result given by the proposed deterministic 
algorithm as a systematic difference between the two compared therapy groups might exist or measuring
inaccuracies for continuous covariates might be too large in the given dataset to draw reliable conclusions.

Finally, we highlight that the algorithm can be used 
as an a priori method for another matching method to extract all available information contained in 
exact matchings, therefore ascertaining that at least the exact matchable 
patients of both therapy groups are matched deterministically and their 
information is completely used. Further research will be dedicated to extend 
the presented model towards $\delta$-matching for $\delta > 0$ while keeping the
desirable properties presented in this paper and therefore extend the 
applicability of the proposed method.

\bibliographystyle{chicago}
\bibliography{biblio}

\begin{thebibliography}{}

\bibitem[\protect\citeauthoryear{Abidov, Rozanski, Hachamovitch, Hayes,
  Aboul-Enein, Cohen, Friedman, Germano, and Berman}{Abidov
  et~al.}{2005}]{Abidov2005}
Abidov, A., A.~Rozanski, R.~Hachamovitch, S.~W. Hayes, F.~Aboul-Enein,
  I.~Cohen, J.~D. Friedman, G.~Germano, and D.~S. Berman (2005).
\newblock Prognostic significance of dyspnea in patients referred for cardiac
  stress testing.
\newblock {\em New England Journal of Medicine\/}~{\em 353\/}(18), 1889--1898.
\newblock PMID: 16267320.

\bibitem[\protect\citeauthoryear{Adams, Gibbons, and Tudehope}{Adams
  et~al.}{2017}]{Adams2017}
Adams, N., K.~S. Gibbons, and D.~Tudehope (2017, Apr).
\newblock Public-private differences in short-term neonatal outcomes following
  birth by prelabour caesarean section at early and full term.
\newblock {\em The Australian \& New Zealand journal of obstetrics \&
  gynaecology\/}~{\em 57}, 176--185.

\bibitem[\protect\citeauthoryear{Anderson, Kish, and Cornell}{Anderson
  et~al.}{1980}]{Anderson1980}
Anderson, D.~W., L.~Kish, and R.~G. Cornell (1980).
\newblock On stratification, grouping and matching.
\newblock {\em Scandinavian Journal of Statistics\/}~{\em 7\/}(2), 61--66.

\bibitem[\protect\citeauthoryear{Austin}{Austin}{2011}]{Austin2011}
Austin, P.~C. (2011, May).
\newblock An introduction to propensity score methods for reducing the effects
  of confounding in observational studies.
\newblock {\em Multivariate behavioral research\/}~{\em 46}, 399--424.

\bibitem[\protect\citeauthoryear{Austin and Small}{Austin and
  Small}{2014}]{Austin2014}
Austin, P.~C. and D.~S. Small (2014).
\newblock The use of bootstrapping when using propensity-score matching without
  replacement: a simulation study.
\newblock {\em Statistics in Medicine\/}~{\em 33\/}(24), 4306--4319.

\bibitem[\protect\citeauthoryear{Burden, Roche, Miglio, Hillyer, Postma,
  Herings, Overbeek, Khalid, van Eickels, and Price}{Burden
  et~al.}{2017}]{Burden2017}
Burden, A., N.~Roche, C.~Miglio, E.~V. Hillyer, D.~S. Postma, R.~M. Herings,
  J.~A. Overbeek, J.~M. Khalid, D.~van Eickels, and D.~B. Price (2017).
\newblock An evaluation of exact matching and propensity score methods as
  applied in a comparative effectiveness study of inhaled corticosteroids in
  asthma.
\newblock {\em Pragmatic and observational research\/}~{\em 8}, 15--30.

\bibitem[\protect\citeauthoryear{Capucci, De~Simone, Luzi, Calvi, Stabile,
  D'Onofrio, Maffei, Leoni, Morani, Sangiuolo, Amellone, Checchinato,
  Ammendola, and Buja}{Capucci et~al.}{2017}]{Capucci2017}
Capucci, A., A.~De~Simone, M.~Luzi, V.~Calvi, G.~Stabile, A.~D'Onofrio,
  S.~Maffei, L.~Leoni, G.~Morani, R.~Sangiuolo, C.~Amellone, C.~Checchinato,
  E.~Ammendola, and G.~Buja (2017, Sep).
\newblock Economic impact of remote monitoring after implantable defibrillators
  implantation in heart failure patients: an analysis from the effect study.
\newblock {\em Europace : European pacing, arrhythmias, and cardiac
  electrophysiology : journal of the working groups on cardiac pacing,
  arrhythmias, and cardiac cellular electrophysiology of the European Society
  of Cardiology\/}~{\em 19}, 1493--1499.

\bibitem[\protect\citeauthoryear{Chen, Wang, Xu, Yan, Gao, Lu, and Wang}{Chen
  et~al.}{2016}]{Chen2016}
Chen, H.-Y., Q.~Wang, Q.-H. Xu, L.~Yan, X.-F. Gao, Y.-H. Lu, and L.~Wang
  (2016).
\newblock Statin as a combined therapy for advanced-stage ovarian cancer: A
  propensity score matched analysis.
\newblock {\em BioMed research international\/}~{\em 2016}, 9125238.

\bibitem[\protect\citeauthoryear{Cho, Choi, Kim, Seo, Kim, Kim, Shin, Lee,
  Ryeom, and Kim}{Cho et~al.}{2017}]{Cho2017}
Cho, S.~H., G.-S. Choi, G.~C. Kim, A.~N. Seo, H.~J. Kim, W.~H. Kim, K.-M. Shin,
  S.~M. Lee, H.~Ryeom, and S.~H. Kim (2017, Mar).
\newblock Long-term outcomes of surgery alone versus surgery following
  preoperative chemoradiotherapy for early t3 rectal cancer: A propensity score
  analysis.
\newblock {\em Medicine\/}~{\em 96}, e6362.

\bibitem[\protect\citeauthoryear{Dou, Yu, Yang, Cheng, Han, Liu, Yu, and
  Liang}{Dou et~al.}{2017}]{Dou2017}
Dou, J.-P., J.~Yu, X.-H. Yang, Z.-G. Cheng, Z.-Y. Han, F.-Y. Liu, X.-L. Yu, and
  P.~Liang (2017, Apr).
\newblock Outcomes of microwave ablation for hepatocellular carcinoma adjacent
  to large vessels: a propensity score analysis.
\newblock {\em Oncotarget\/}~{\em 8}, 28758--28768.

\bibitem[\protect\citeauthoryear{Fukami, Takeuchi, Kagaya, Ojima, Saito, Sato,
  Matsuda, and Nagasawa}{Fukami et~al.}{2017}]{Fukami2017}
Fukami, H., Y.~Takeuchi, S.~Kagaya, Y.~Ojima, A.~Saito, H.~Sato, K.~Matsuda,
  and T.~Nagasawa (2017).
\newblock Perirenal fat stranding is not a powerful diagnostic tool for acute
  pyelonephritis.
\newblock {\em International journal of general medicine\/}~{\em 10}, 137--144.

\bibitem[\protect\citeauthoryear{Gozalo, Plotzke, Mor, Miller, and Teno}{Gozalo
  et~al.}{2015}]{Gozalo2015}
Gozalo, P., M.~Plotzke, V.~Mor, S.~C. Miller, and J.~M. Teno (2015).
\newblock Changes in medicare costs with the growth of hospice care in nursing
  homes.
\newblock {\em New England Journal of Medicine\/}~{\em 372\/}(19), 1823--1831.
\newblock PMID: 25946281.

\bibitem[\protect\citeauthoryear{Hansen}{Hansen}{2004}]{Hansen2004}
Hansen, B. (2004, 02).
\newblock Full matching in an observational study of coaching for the sat.
\newblock {\em Journal of the American Statistical Association\/}~{\em 99},
  609--618.

\bibitem[\protect\citeauthoryear{Hansen and Klopfer}{Hansen and
  Klopfer}{2012}]{Hansen2012}
Hansen, B. and S.~Klopfer (2012, 01).
\newblock Optimal full matching and related designs via network flows.
\newblock {\em Journal of Computational and Graphical Statistics\/}~{\em 15}.

\bibitem[\protect\citeauthoryear{Iacus, King, Porro, and Katz}{Iacus
  et~al.}{2012}]{Iacus2012}
Iacus, S., G.~King, G.~Porro, and J.~Katz (2012, 12).
\newblock Causal inference without balance checking: Coarsened exact matching.
\newblock {\em Political Analysis\/}~{\em 20}, 1--24.

\bibitem[\protect\citeauthoryear{King and Nielsen}{King and
  Nielsen}{2019}]{King2019}
King, G. and R.~Nielsen (2019).
\newblock Why propensity scores should not be used for matching.
\newblock {\em Political Analysis\/}, 1–20.

\bibitem[\protect\citeauthoryear{Kishimoto, Yamana, Inoue, Noda, Myojin,
  Matsui, Yasunaga, Kawaguchi, and Imamura}{Kishimoto
  et~al.}{2017}]{Kishimoto2017}
Kishimoto, M., H.~Yamana, S.~Inoue, T.~Noda, T.~Myojin, H.~Matsui, H.~Yasunaga,
  M.~Kawaguchi, and T.~Imamura (2017, Jun).
\newblock Sivelestat sodium and mortality in pneumonia patients requiring
  mechanical ventilation: propensity score analysis of a japanese nationwide
  database.
\newblock {\em Journal of anesthesia\/}~{\em 31}, 405--412.

\bibitem[\protect\citeauthoryear{Kong, Li, Li, Jiang, Yang, and Yan}{Kong
  et~al.}{2017}]{Kong2017}
Kong, L., M.~Li, L.~Li, L.~Jiang, J.~Yang, and L.~Yan (2017, Apr).
\newblock Splenectomy before adult liver transplantation: a retrospective
  study.
\newblock {\em BMC surgery\/}~{\em 17}, 44.

\bibitem[\protect\citeauthoryear{Kupper, Karon, Kleinbaum, Morgenstern, and
  Lewis}{Kupper et~al.}{1981}]{Kupper1981}
Kupper, L.~L., J.~M. Karon, D.~G. Kleinbaum, H.~Morgenstern, and D.~K. Lewis
  (1981).
\newblock Matching in epidemiologic studies: Validity and efficiency
  considerations.
\newblock {\em Biometrics\/}~{\em 37\/}(2), 271--291.

\bibitem[\protect\citeauthoryear{Lai, Rau, Wu, Chen, Kuo, Hsu, Hsieh, and
  Hsieh}{Lai et~al.}{2016}]{Lai2016}
Lai, W.-H., C.-S. Rau, S.-C. Wu, Y.-C. Chen, P.-J. Kuo, S.-Y. Hsu, C.-H. Hsieh,
  and H.-Y. Hsieh (2016, Nov).
\newblock Post-traumatic acute kidney injury: a cross-sectional study of trauma
  patients.
\newblock {\em Scandinavian journal of trauma, resuscitation and emergency
  medicine\/}~{\em 24}, 136.

\bibitem[\protect\citeauthoryear{Lee, Lee, Kim, Choo, Chung, Lee, and Jung}{Lee
  et~al.}{2017}]{Lee2017}
Lee, S.~I., K.~S. Lee, J.~B. Kim, S.~J. Choo, C.~H. Chung, J.~W. Lee, and S.-H.
  Jung (2017, Jun).
\newblock Early antithrombotic therapy after bioprosthetic aortic valve
  replacement in elderly patients: A single-center experience.
\newblock {\em Annals of thoracic and cardiovascular surgery : official journal
  of the Association of Thoracic and Cardiovascular Surgeons of Asia\/}~{\em
  23}, 128--134.

\bibitem[\protect\citeauthoryear{Liu, Han, Liu, Yang, Jiang, and Wang}{Liu
  et~al.}{2016}]{Liu2016}
Liu, Y., J.~Han, T.~Liu, Z.~Yang, H.~Jiang, and H.~Wang (2016).
\newblock The effects of diabetes mellitus in patients undergoing off-pump
  coronary artery bypass grafting.
\newblock {\em BioMed research international\/}~{\em 2016}, 4967275.

\bibitem[\protect\citeauthoryear{McDonald, McDonald, Williamson, Kallmes, and
  Kashani}{McDonald et~al.}{2017}]{McDonald2017}
McDonald, J.~S., R.~J. McDonald, E.~E. Williamson, D.~F. Kallmes, and
  K.~Kashani (2017, Jun).
\newblock Post-contrast acute kidney injury in intensive care unit patients: a
  propensity score-adjusted study.
\newblock {\em Intensive care medicine\/}~{\em 43}, 774--784.

\bibitem[\protect\citeauthoryear{McEvoy, Antic, Heeley, Luo, Ou, Zhang,
  Mediano, Chen, Drager, Liu, Chen, Du, McArdle, Mukherjee, Tripathi, Billot,
  Li, Lorenzi-Filho, Barbe, Redline, Wang, Arima, Neal, White, Grunstein,
  Zhong, and Anderson}{McEvoy et~al.}{2016}]{McEvoy2016}
McEvoy, R.~D., N.~A. Antic, E.~Heeley, Y.~Luo, Q.~Ou, X.~Zhang, O.~Mediano,
  R.~Chen, L.~F. Drager, Z.~Liu, G.~Chen, B.~Du, N.~McArdle, S.~Mukherjee,
  M.~Tripathi, L.~Billot, Q.~Li, G.~Lorenzi-Filho, F.~Barbe, S.~Redline,
  J.~Wang, H.~Arima, B.~Neal, D.~P. White, R.~R. Grunstein, N.~Zhong, and C.~S.
  Anderson (2016).
\newblock Cpap for prevention of cardiovascular events in obstructive sleep
  apnea.
\newblock {\em New England Journal of Medicine\/}~{\em 375\/}(10), 919--931.
\newblock PMID: 27571048.

\bibitem[\protect\citeauthoryear{Pearl}{Pearl}{2000}]{Pearl2000}
Pearl, J. (2000, 01).
\newblock Causality: Models, reasoning, and inference, second edition.
\newblock {\em Causality\/}~{\em 29}.

\bibitem[\protect\citeauthoryear{Ray, Murray, Hall, Arbogast, and Stein}{Ray
  et~al.}{2012}]{Ray2012}
Ray, W., K.~Murray, K.~Hall, P.~Arbogast, and C.~Stein (2012, 05).
\newblock Azithromycin and the risk of cardiovascular death reply.
\newblock {\em The New England journal of medicine\/}~{\em 366}, 1881--90.

\bibitem[\protect\citeauthoryear{Rosenbaum and Rubin}{Rosenbaum and
  Rubin}{1983}]{RubinRosenbaum1983}
Rosenbaum, P. and D.~Rubin (1983, 04).
\newblock The central role of the propensity score in observational studies for
  causal effects.
\newblock {\em Biometrika\/}~{\em 70}, 41--55.

\bibitem[\protect\citeauthoryear{Rosenbaum}{Rosenbaum}{1989}]{Rosenbaum1989}
Rosenbaum, P.~R. (1989).
\newblock Optimal matching for observational studies.
\newblock {\em Journal of the American Statistical Association\/}~{\em
  84\/}(408), 1024--1032.

\bibitem[\protect\citeauthoryear{Rubin}{Rubin}{1972}]{Rubin1972}
Rubin, D. (1972).
\newblock Estimating causal effects of treatments in experimental and
  observational studies.
\newblock {\em ETS Research Bulletin Series\/}~{\em 1972\/}(2), i--31.

\bibitem[\protect\citeauthoryear{Rubin}{Rubin}{1973}]{Rubin1973}
Rubin, D.~B. (1973).
\newblock Matching to remove bias in observational studies.
\newblock {\em Biometrics\/}~{\em 29\/}(1), 159--183.

\bibitem[\protect\citeauthoryear{Salati, Brunelli, Xiume, Monteverde,
  Sabbatini, Tiberi, Pompili, Palloni, and Refai}{Salati
  et~al.}{2017}]{Salati2017}
Salati, M., A.~Brunelli, F.~Xiume, M.~Monteverde, A.~Sabbatini, M.~Tiberi,
  C.~Pompili, R.~Palloni, and M.~Refai (2017, Jun).
\newblock Video-assisted thoracic surgery lobectomy does not offer any
  functional recovery advantage in comparison to the open approach 3 months
  after the operation: a case matched analysisdagger.
\newblock {\em European journal of cardio-thoracic surgery : official journal
  of the European Association for Cardio-thoracic Surgery\/}~{\em 51},
  1177--1182.

\bibitem[\protect\citeauthoryear{Schermerhorn, O'Malley, Jhaveri, Cotterill,
  Pomposelli, and Landon}{Schermerhorn et~al.}{2008}]{Schermerhorn2008}
Schermerhorn, M.~L., A.~J. O'Malley, A.~Jhaveri, P.~Cotterill, F.~Pomposelli,
  and B.~E. Landon (2008).
\newblock Endovascular vs. open repair of abdominal aortic aneurysms in the
  medicare population.
\newblock {\em New England Journal of Medicine\/}~{\em 358\/}(5), 464--474.
\newblock PMID: 18234751.

\bibitem[\protect\citeauthoryear{Seung, Park, Kim, Lee, Lee, Hong, Park, Yun,
  Gwon, Jeong, Jang, Kim, Kim, Seong, Park, Ahn, Chae, Tahk, Chung, and
  Park}{Seung et~al.}{2008}]{Seung2008}
Seung, K.~B., D.-W. Park, Y.-H. Kim, S.-W. Lee, C.~W. Lee, M.-K. Hong, S.-W.
  Park, S.-C. Yun, H.-C. Gwon, M.-H. Jeong, Y.~Jang, H.-S. Kim, P.~J. Kim,
  I.-W. Seong, H.~S. Park, T.~Ahn, I.-H. Chae, S.-J. Tahk, W.-S. Chung, and
  S.-J. Park (2008).
\newblock Stents versus coronary-artery bypass grafting for left main coronary
  artery disease.
\newblock {\em New England Journal of Medicine\/}~{\em 358\/}(17), 1781--1792.
\newblock PMID: 18378517.

\bibitem[\protect\citeauthoryear{Shaw, Stafford-Smith, White, Phillips-Bute,
  Swaminathan, Milano, Welsby, Aronson, Mathew, Peterson, and Newman}{Shaw
  et~al.}{2008}]{Shaw2008}
Shaw, A.~D., M.~Stafford-Smith, W.~D. White, B.~Phillips-Bute, M.~Swaminathan,
  C.~Milano, I.~J. Welsby, S.~Aronson, J.~P. Mathew, E.~D. Peterson, and M.~F.
  Newman (2008, Feb).
\newblock The effect of aprotinin on outcome after coronary-artery bypass
  grafting.
\newblock {\em The New England journal of medicine\/}~{\em 358}, 784--93.

\bibitem[\protect\citeauthoryear{Stuart}{Stuart}{2010}]{Stuart2010}
Stuart, E.~A. (2010, Feb).
\newblock Matching methods for causal inference: A review and a look forward.
\newblock {\em Statistical science : a review journal of the Institute of
  Mathematical Statistics\/}~{\em 25}, 1--21.

\bibitem[\protect\citeauthoryear{Svanström, Pasternak, and Hviid}{Svanström
  et~al.}{2013}]{Svanstroem2013}
Svanström, H., B.~Pasternak, and A.~Hviid (2013).
\newblock Use of azithromycin and death from cardiovascular causes.
\newblock {\em New England Journal of Medicine\/}~{\em 368\/}(18), 1704--1712.
\newblock PMID: 23635050.

\bibitem[\protect\citeauthoryear{Tranchart, Fuks, Vigano, Ferretti, Paye,
  Wakabayashi, Ferrero, Gayet, and Dagher}{Tranchart
  et~al.}{2016}]{Tranchart2016}
Tranchart, H., D.~Fuks, L.~Vigano, S.~Ferretti, F.~Paye, G.~Wakabayashi,
  A.~Ferrero, B.~Gayet, and I.~Dagher (2016, May).
\newblock Laparoscopic simultaneous resection of colorectal primary tumor and
  liver metastases: a propensity score matching analysis.
\newblock {\em Surgical endoscopy\/}~{\em 30}, 1853--62.

\bibitem[\protect\citeauthoryear{Zangbar, Khalil, Gruessner, Joseph, Friese,
  Kulvatunyou, Wynne, Latifi, Rhee, and O'Keeffe}{Zangbar
  et~al.}{2016}]{Zangbar2016}
Zangbar, B., M.~Khalil, A.~Gruessner, B.~Joseph, R.~Friese, N.~Kulvatunyou,
  J.~Wynne, R.~Latifi, P.~Rhee, and T.~O'Keeffe (2016, Nov).
\newblock Levetiracetam prophylaxis for post-traumatic brain injury seizures is
  ineffective: A propensity score analysis.
\newblock {\em World journal of surgery\/}~{\em 40}, 2667--2672.

\bibitem[\protect\citeauthoryear{Zhang, Guddeti, Matsuzawa, Sara, Kwon, yue
  Liu, Sun, Lee, Lennon, Bell, Schaff, Daly, Lerman, Lerman, and Locker}{Zhang
  et~al.}{2016}]{Zhang2016LeftIM}
Zhang, M., R.~R. Guddeti, Y.~Matsuzawa, J.~D. Sara, T.~Kwon, Z.~yue Liu,
  T.~Sun, S.~Lee, R.~J. Lennon, M.~R. Bell, H.~V. Schaff, R.~C. Daly, L.~O.
  Lerman, A.~Lerman, and C.~Locker (2016).
\newblock Left internal mammary artery versus coronary stents: Impact on
  downstream coronary stenoses and conduit patency.
\newblock In {\em Journal of the American Heart Association}.

\bibitem[\protect\citeauthoryear{Zhang, Chen, and Ni}{Zhang
  et~al.}{2015}]{Zhang2015}
Zhang, Z., K.~Chen, and H.~Ni (2015).
\newblock Calcium supplementation improves clinical outcome in intensive care
  unit patients: a propensity score matched analysis of a large clinical
  database mimic-ii.
\newblock {\em SpringerPlus\/}~{\em 4}, 594.

\end{thebibliography}

\end{document}